\documentclass[11pt,twocolumn]{article}

\usepackage[utf8]{inputenc}

\usepackage{mathtools}

\usepackage{graphicx}
\usepackage{epstopdf}

\usepackage{cite}
\usepackage{url}

\usepackage{amsmath,amssymb,amsfonts,amsthm}

\usepackage{authblk}

\usepackage{enumerate}

\newtheorem{theorem}{Theorem}[section]

\newtheorem{proposition}[theorem]{Proposition}

\newtheorem{definition}[theorem]{Definition}

\usepackage{xcolor}
\definecolor{myOrange}{HTML}{E65C0D}
\definecolor{myBlue}{HTML}{5D6DC2}
\usepackage{algorithm}
\usepackage{listings}
\definecolor{colortodobg}{cmyk}{0,0,0.2,0}

\newcommand{\orth}{\ensuremath{\mathcal{MO}}}

\newcount\colveccount
\newcommand*\colvec[1]{
        \global\colveccount#1
        \begin{pmatrix}
        \colvecnext
}
\def\colvecnext#1{
        #1
        \global\advance\colveccount-1
        \ifnum\colveccount>0
                \\
                \expandafter\colvecnext
        \else
                \end{pmatrix}
        \fi
}

\newcommand{\ram}{\ensuremath{\texttt{RAM}}{}}

\newcommand{\dtw}{\ensuremath{\texttt{DTW}}{}}
\newcommand{\cd}{\ensuremath{\texttt{CD}}{}}
\newcommand{\frechet}{\ensuremath{\texttt{DK}}{}}

\newcommand{\NN}{\ensuremath{\mathbb{N}}}
\newcommand{\RR}{\ensuremath{\mathbb{R}}}

\renewcommand{\geq}{\geqslant}
\renewcommand{\leq}{\leqslant}

\newcommand{\deff}{\coloneqq}

\newcommand{\dE}{\ensuremath{\textsf{\upshape d}}}

\renewcommand{\phi}{\ensuremath{\varphi}}
\renewcommand{\epsilon}{\ensuremath{\varepsilon}}

\newcommand{\ourcomments}[1]{
  \textbf{\textcolor{green}{#1}}
}

\renewcommand{\ourcomments}[1]{}

\begin{document}

\title{Efficient Measuring of Congruence on High Dimensional Time Series}

\author[1]{Jörg P. Bachmann}
\author[2]{Johann-Christoph Freytag}
\affil[1]{ \texttt{joerg.bachmann@informatik.hu-berlin.de}}
\affil[2]{ \texttt{freytag@informatik.hu-berlin.de}}
\affil[1,2]{Humboldt-Universität zu Berlin, Germany}

\date{\today}

\maketitle

\begin{abstract}
    A time series is a sequence of data items; typical examples are streams of temperature measurements, stock ticker data, or gestures recorded with modern virtual reality motion controllers.
    Quite some research has been devoted to comparing and indexing time series.
    Especially, when the comparison should not be affected by time warping, the ubiquitous Dynamic Time Warping distance function (\dtw) is one of the most analyzed time series distance functions.
    The Dog-Keeper distance (\frechet) is another example for a distance function on time series which is truely invariant under time warping.

    For many application scenarios (e.\,g. motion gesture recognition in virtual reality), the invariance under isometric spatial transformations (i.\,e. rotation, translation, and mirroring) is as important as the invariance under time warping.
    Distance functions on time series which are invariant under isometric transformations can be seen as measurements for the congruency of two time series.
    The congruence distance (\cd{}) is an example for such a distance function.
    However, it is very hard to compute and it is not invariant under time warpings.

    In this work, we are taking one step towards developing a feasable distance function which is invariant under isometric spatial transformations and time warping:
    We develop four approximations for \cd{}.
    Two of these even satisfy the triangle inequality and can thus be used with metric indexing structures.
    We show that all approximations serve as a lower bound to \cd{}.
    Our evaluation shows that they achieve remarkable tightness while providing a speedup of more than two orders of magnitude to the congruence distance.
\end{abstract}

\section{Introduction}\label{sec:introduction}

Multimedia retrieval is a common application which requires finding similar objects to a query object.
We consider examples such as gesture recognition with modern virtual reality motion controllers and classification of handwritten letters where the objects are multi-dimensional time series.

In many cases, similarity search is performed using a distance function on the time series, where small distances imply similar time series.
A \emph{nearest neigbor query} to the query time series can be a $k$-nearest neighbor ($k$-NN) query or an $\epsilon$-nearest neighbor ($\epsilon$-NN) query:
A $k$-NN query retrieves the $k$ most similar time series;
an $\epsilon$-NN query retrieves all time series with a distance of at most $\epsilon$.

In our examples, the time series of the same classes (e.\,g., same written characters or same gestures) differ by temporal as well as spatial displacements.
\emph{Time warping distance functions} such as dynamic time warping (\dtw)~\cite{DTWSakoe} and the Dog-Keeper distance (\frechet)~\cite{WDK17,computingfrechet} are robust against temporal displacements.
They map pairs of time series representing the same trajectory to small distances.
Still, they fail when the time series are rotated or translated in space.

The distance functions defined and analyzed in this paper measure the (approximate) congruence of two time series.
Thereby, the distance between two time series $S$ and $T$ shall be $0$ iff $S$ can be transformed into $T$ by rotation, translation, and mirroring; in this case, $S$ and $T$ are said to be \emph{congruent}.
A value greater than $0$ shall correlate to the amount of transformation needed to turn the time series into congruent ones.

The classical \textsc{Congruence} problem basically determines whether two point sets $A,B\subseteq\mathbb R^k$ are congruent considering isometric transformations (i.\,e., rotation, translation, and mirroring) \cite{Heffernan:1992:ADA:142675.142697,Alt:1988:CSS:44611.44614}.
For $2$- and $3$-dimensional spaces, there are results providing algorithms with runtime $\mathcal O(n\cdot\log n)$ when $n$ is the size of the sets \cite{Alt:1988:CSS:44611.44614}.
For larger dimensionalities, they provide an algorithm with runtime $\mathcal O(n^{k-2}\cdot\log n)$.
For various reasons (e.\,g. bounded floating point precision, physical measurement errors), the \emph{approximated} \textsc{Congruence} problem is of much more interest in practical applications.
Different variations of the approximated \textsc{Congruence} problem have been studied (e.\,g. what types of transformations are used, is the assignment of points from $A$ to $B$ known, what metric is used) \cite{Heffernan:1992:ADA:142675.142697,Alt:1988:CSS:44611.44614,Indyk:2003:ACN:636968.636974,Alt96discretegeometric}.

The \textsc{Congruence} problem is related to our work, since the problem is concerned with the existence of isometric functions such that a point set maps to another point set.
The main difference is, that we consider ordered lists of points (i.\,e. time series) rather than pure sets.
It turned out, that solving the approximated \textsc{Congruence} problem is NP-hard regarding length and dimensionality \cite{Congruence}.


With this work, we contribute by evaluating the congruence distance with an implementation based on a nonlinear optimizer.
We propose two approximations to the congruence distance which have linear runtime regarding the dimensionality and (quasi-) quadratic runtime regarding the length of the time series.
We improve the complexity of both approximations at cost of approximation quality, such that their complexity is (quasi-) linear regarding the length of the time series.
We evaluate the approximations experimentally.

\subsection{Basic Notation}
\label{sec:notation}

We denote the natural numbers including zero with $\NN$ and the real numbers with $\RR$.
For $a,b\in\NN$ we denote the modulo operator by $a\%b$.
The set of all powers of two is denoted via $2^\NN\coloneqq\left\{ 1,2,4,8,\cdots \right\}$.

Elements of a $k$-dimensional vector $v\in\RR^k$ are accessed using subindices, i.\,e. $v_3$ is the third element of the vector.
Sequences (here also called time series) are usually written using capital letters, e.\,g. $S=(s_0,\cdots,s_{n-1})$ is a sequence of length $n$.
Suppose $s_i\in\RR^k$, then $s_{i,j}$ denotes the $j$-th element of the $i$-th vector in the sequence $S$.
The projection to the $j$-th dimension is denoted via $S^j$, i.\,e. $S^j=(s_{0,j},\cdots,s_{n-1,j})$.
The Euclidean norm of a vector $v$ is denoted via $\|v\|_2$, thus $\dE(v,w)\coloneqq\|v-w\|_2$ denotes the Euclidean distance between $v$ and $w\in\RR^k$.

We denote the set of $k$-dimensional orthogonal matrices with $\orth(k)$, the identity matrix with $I$ and the transposed of a matrix $M$ with $M^T$.
For a matrix $M$ in $\RR^k$, we denote the matrix holding the absoloute values with $|M|\coloneqq \left( |M_{i,j}| \right)_{0\leq i,j < k}$.

\subsection{Congruence Distance}

While \dtw{} compares two time series $S$ and $T$, it is (nearly) invariant under time warpings.
In detail, consider $\sigma(S)$ and $\tau(T)$ as warpings by duplicating elements (e.\,g. $\sigma(S)=\left( s_0,s_1,s_1,s_1,s_2,s_3,s_3,s_4,\cdots \right)$), then $\dtw$ minimizes the L1 distance under all time warps:
\begin{align*}
    \dtw(S,T) \coloneqq \min_{\sigma,\tau} \sum_{i=0}^{|\sigma(S)|-1} \dE\left( \sigma(S)_i, \tau(T)_i\right)
\end{align*}
with $|\sigma(S)|=|\tau(T)|$.

On the other hand, the congruence distance is invariant under all isometric transformations.
The difference to \dtw{} is, that it minimizes the L1 distance by multiplying an orthogonal matrix and adding a vector:
\begin{align}
    \label{eq:congruence}
    d^C(S,T) &\coloneqq \min_{M,v} f(M,v) \nonumber \\
    &\coloneqq \min_{M,v} \sum_{i=0}^{n-1} \dE\left( s_i, M\cdot t_i+v\right)
\end{align}
The computation of $d^C(S,T)$ is an optimization problem where $f$ (cf. Equation~(\ref{eq:congruence})) corresponds to the objective function.
For time series in $\RR^k$, the orthogonality of $M$ yields a set of $k^2$ equality based constraints.

\section{Approximating the Congruency}
\label{sec:approximation}

Consider two time series $S,T$, an arbitrary orthogonal matrix $M\in\orth(k)$, and a vector $v\in\RR^k$.
Using the triangle inequality, we obtain
\begin{align*}
    \label{eq:delta}
    &\dE(s_i, M\cdot s_j+v) \leq \\
    &\, \dE(s_i, M\cdot t_i+v)+\dE(t_i, t_j)+\dE(M\cdot t_j+v, s_j) \\
    \Rightarrow &\left| \dE(s_i, s_j)-\dE(t_i, t_j)\right| \leq \\
    &\,\dE(s_i, M\cdot t_i+v)+\dE(s_j, M\cdot t_j+v) \addtocounter{equation}{1}\tag{\theequation}
\end{align*}
i.\,e. we can estimate the congruence distance $\sum_{i=0}^{n-1} \dE(s_i, M\cdot t_i+v)$ without actually solving the optimization problem.
We unroll this idea to propose two approximating algorithms in Section~\ref{sec:approxdelta}~and~\ref{sec:approxgreedy}.

Considering the well-known \emph{self-similarity matrix} of a time series, the left hand side of Equation~(\ref{eq:delta}) matches the entry of the difference between two self-similarity matrices.
Usually, the self-similarity matrix is used to analyze a time series for patterns (e.\,g. using Recurrence Plots \cite{recurrenceplots}).
The important property that makes the self-similarity matrix useful for approximating the congruence distance, is its invariance under transformations considered for the congruence distance, i.\,e. rotation, translation, and mirroring.

The \emph{self-similarity matrix} of an arbitrary time series $T=(t_0,\dots,t_{n-1})$ is defined as follows:
\begin{displaymath}
    \Delta T \ \ \coloneqq \ \ \big(\; \dE\left(t_i,t_j\right)\; \big)_{0\leq i,j < n}
\end{displaymath}
Note, that $\Delta T_{i,j} = \Delta T_{j,i}$ and $\Delta T_{i,i}=0$.
In fact, the self-similarity matrix $\Delta T$ completely describes the sequence $T$ up to congruence, i.e., up to rotation, translation, and mirroring of the whole sequence in $\RR^k$ \cite{Congruence}:
Two time series $S$ and $T$ are congruent iff they have the same self-similarity matrix, i.\,e.
\begin{align}
    \label{eq:congiffdelta}
    &\exists M\in\orth(k),v\in\RR^k: \nonumber \\
    &\quad S=M\cdot T+v \ \iff \ \Delta S=\Delta T.
\end{align}

\subsection{Metric Approximation}
\label{sec:approxdelta}

Equation~(\ref{eq:delta})~and~(\ref{eq:congiffdelta}) yield the approach for approximating the congruency:
We measure the congruency of two time series $S$ and $T$ via a metric on their self-similarity matrices.

\begin{definition}[Delta Distance]
    \label{def:deltadistance}
    Let $S,T$ be two time series of length $n$.
    The \emph{delta distance} $d^\Delta(S,T)$ is defined as follows: 
    \begin{align*}
        &d^\Delta(S,T) \coloneqq \\
        &\quad \frac 1 2 \max_{0<\delta<n} \sum_{i=0}^{n-1} \left| \dE\left(s_i,s_{(i+\delta)\% n}\right) - \dE\left( t_i,t_{(i+\delta)\% n} \right) \right|
    \end{align*}
\end{definition}

\begin{proposition}
    The delta distance satisfies the triangle inequality.
\end{proposition}
\begin{proof}
    Consider three time series $R,S,$ and $T$ and fixiate a $\delta^*$ which maximizes $d^\Delta(R,T)$ in Definition~\ref{def:deltadistance}.
    Then
    \begin{align*}
        &d^\delta(R,T) = \\
        &\quad\sum_{i=0}^{n-1} \left| \dE\left( r_i,r_{(i+\delta^*)\% n} \right)-\dE\left( t_i,t_{(i+\delta^*)\% n} \right) \right| \\
        &\leq\sum_{i=0}^{n-1} \left| \dE\left( r_i,r_{(i+\delta^*)\% n} \right)+\dE\left( s_i,s_{(i+\delta^*)\% n} \right) \right| + \\
        &\quad\left| \dE\left( s_i,s_{(i+\delta^*)\% n} \right)-\dE\left( t_i,t_{(i+\delta^*)\% n} \right) \right| \\
        &\leq d^\Delta(R,S)+d^\Delta(S,T)
    \end{align*}
    prooves the triangle inequality.
\end{proof}
Since $\dE$ is symmetric, $d^\Delta$ inherits its symmetry.
Hence, $d^\Delta$ is a pseudo metric on the set of time series of length $n$ where all time series of an equivalence class are congruent to each other.

We omit providing pseudo code since the computation matches the formula in Definition~\ref{def:deltadistance}.
The complexity of computing the delta distance $d^\Delta$ grows quadratically with the length of the time series.

Our next aim is to show that the the delta distance $d^\Delta$ provides a lower bound on the congruence distance $d^C$, as formulated in the following theorem.

\begin{theorem}
    \label{thm:quadmaxest}
    For all time series $S$ and $T$, the following holds:
    \begin{align*}
        d^\Delta(S,T) \ \ \leq \ \ d^C(S,T).
    \end{align*}
\end{theorem}
\begin{proof}
    Fixiate a $\delta^*$ which maximizes $d^\Delta(S,T)$ in Definition~\ref{def:deltadistance}.
    Using the triangle inequality as in Equation~(\ref{eq:delta}) yields
    \begin{align*}
        &d^\Delta(S,T) = \\
        &\quad\frac 1 2\sum_{i=0}^{n-1} \left| \dE\left( s_i,s_{(i+\delta^*)\% n} \right)-\dE\left( t_i,t_{(i+\delta^*)\% n} \right)\right| \\
        &= \frac 1 2\sum_{i=0}^{n-1} \left| \dE\left( s_i,s_{(i+\delta^*)\% n} \right)- \right. \\
        &\quad\quad \left. \dE\left( M\cdot t_i+v,M\cdot t_{(i+\delta^*)\% n}+v \right)\right| \\
        &\leq \frac 1 2\sum_{i=0}^{n-1} \left( \dE\left(s_i,M\cdot t_i+v\right)+ \right. \\
        &\quad\quad \left. \dE\left(s_{(i+\delta^*)\% n}, M\cdot t_{(i+\delta^*)\% n}+v \right) \right) \\
        &= \sum_{i=0}^{n-1} \dE\left( s_i,M\cdot t_i+v \right)
    \end{align*}
    for arbitrary $M\in\orth(k)$ and $v\in\RR^k$.
    Hence, $d^\Delta(S,T)\leq d^C(S,T)$.
\end{proof}

In this section, we provided the delta distance, which is a metric lower bound to the congruence distance.

\subsection{Greedy Approximation}
\label{sec:approxgreedy}

The approach of the delta distance is simple:
For time series $S$ and $T$, it only sums up values along a (wrapped) diagonal in $|\Delta S-\Delta T|$ and chooses the largest value.
However, another combination of elements within $|\Delta S-\Delta T|$ as addends might provide a better approximation of the congruence distance.
Since it is a computational expensive task, to try all combinations, we try to find a good combination using a greedy algorithm for selecting the entries of $|\Delta S-\Delta T|$.

The greedy algorithm first sorts the elements $d_{i,j} = \left|\dE(s_i,s_j)-\dE(t_i,t_j)\right|$ in descending order and stores them in a sequence $Q=\left( d_{i_1,j_1}, d_{i_2,j_2},\cdots \right)$.
While iterating over the sequence $Q$, it adds $d_{i_r,j_r}$ to a global sum and masks the indices $i_r$ and $j_r$ as already seen.
Elements in the queue which access already seen indices are skipped, thus each index is used at most once.
Basically, this is the reason, why the greedy delta distance (denoted as $d^G(S,T)$) is a lower bound to the congruence distance.
Theorem~\ref{thm:greedylb} proves the last statement and Algorithm~\ref{alg:approxgreedy} provides the pseudo code for the computation.

\begin{algorithm}
    \caption{Greedy Delta Distance}
    \begin{lstlisting}
Algorithm: greedydelta
Input: time series $S, T$ of length $n$
Output: distance $d$

let $Q=()$ // empty sequence
for $i=0,\ldots,n-2$
  for $j=i+1,\ldots,n-1$
    append $d_{i,j}\coloneqq\left|\dE\left( s_i,s_j \right)-\dE\left( t_i,t_j \right)\right|$ to $Q$
sort $Q$ // (descending)
let $S=\emptyset$
let $d = 0$
for each $d_{i_a,j_a}$ in $Q$
  if $i_a\in S$ or $j_a\in S$ continue
  let $d=d+d_{i_a,j_a}$
  let $S=S\cup \left\{ i_a,j_a \right\}$
return $d$
    \end{lstlisting}
    \label{alg:approxgreedy}
\end{algorithm}

The complexity is dominated by sorting $n^2$ elements, which takes $n^2\cdot\log(n^2)$ steps.

\begin{theorem}
    \label{thm:greedylb}
    For all time series $S$ and $T$, the following holds:
    \begin{align*}
        d^G(S,T) \ \ \leq \ \ d^C(S,T).
    \end{align*}
\end{theorem}
\begin{proof}
    Let $Q^*=\left( d_{i_1,j_1},\cdots,d_{i_r,j_r} \right)$ be the list of elements from the queue in Algorithm~\ref{alg:approxgreedy} which have not been skipped.
    Since each index is appears at most once in this list, the following inequality holds for arbitrary orthogonal matrices $M$ and vectors $v\RR^k$:
    \begin{align*}
        &d^G(S,T) = \sum_{a=1}^r d_{i_a,j_a} \\
        &\leq \sum_{a=1}^r \left| \dE\left( s_{i_a},s_{j_a} \right)-\dE\left( M\cdot t_{i_a}+v,M\cdot t_{j_a}+v \right) \right| \\
        &\leq \sum_{a=1}^r \dE\left( s_{i_a},t_{i_a} \right)+\dE\left( M\cdot t_{i_a}+v,M\cdot t_{j_a}+v \right) \\
        &\leq \sum_{i=0}^{n-1} \dE\left( s_i,M\cdot t_i+v \right)
    \end{align*}
    Hence, $d^G(S,T)\leq d^C(S,T)$.
\end{proof}

\subsection{Runtime improvement}
\label{sec:runtimeimprovement}

The complexity of the delta distance and greedy delta distance is linear regarding the dimensionality but quadratic in length.
In this section, we motivate an optimization for both algorithms.

Time series usually do not contain random points, but they come from continuous processes in the real world, i.\,e. the distance between two successive elements is rather small.
Hence, the distances $\dE\left( t_i,t_j \right)$ and $\dE\left( t_i,t_{j+1} \right)$ are probably close to each other if $i \ll j$, i.\,e. if $j$ is much larger than $i$.
This insight leads to the idea, to only consider elements $\dE\left( t_i,t_j \right)$ where $|i-j|$ is a power of two, i.\,e. we consider less elements for larger temporal distances.

\paragraph*{The Fast Delta Distance:}
Adapting the idea to the delta distance $d^\Delta$ yields the following definition.
\begin{definition}[Fast Delta Distance]
    \label{def:fastdeltadistance}
    Let $S,T$ be two time series of length $n$.
    The \emph{fast delta distance} $\tilde d^\Delta(S,T)$ is defined as follows: 
    \begin{align*}
        &\tilde d^\Delta(S,T) \deff \\
        &\frac 1 2 \max_{0<\delta<n, \delta\in 2^\NN} \sum_{i=0}^{n-1} \left| \dE\left(s_i,s_{(i+\delta)\% n}\right) - \dE\left( t_i,t_{(i+\delta)\% n} \right) \right|
    \end{align*}
\end{definition}
Since we omit some values $\delta$ in Definition~\ref{def:deltadistance}, the fast version $\tilde d^\Delta$ is a lower bound to $d^\Delta$, i.\,e. the following theorem holds:
\begin{theorem}
    \label{thm:fastdeltalb}
    For all time series $S$ and $T$, the following holds:
    \begin{align*}
        \tilde d^\Delta(S,T) \ \ \leq \ \ d^\Delta(S,T)
    \end{align*}
\end{theorem}
Especially, the fast delta distance is also a lower bound to the congruence distance.
For time series of length $n$ the complexity of the fast delta distance $\tilde d^\Delta$ improves to $n\log n$.
On the other hand, equivalence classes regarding the fast delta distance might include time series which are not congruent.

\paragraph*{The Fast Greedy Delta Distance:}
Incorporating the idea for improving the runtime into the greedy delta distance simply changes Line~$7$ of Algorithm~\ref{alg:approxgreedy}:
We only consider values for the variable $j$, which add a power of $2$ to the variable $i$.
Algorithm~\ref{alg:fastgreedy} provides the line to change in Algorithm~\ref{alg:approxgreedy} in order to achieve the fast greedy delta distance.
\begin{algorithm}
    \caption{Distinction between Greedy Delta Distance (cf. Algorithm~\ref{alg:approxgreedy}) and Fast Greedy Delta Distance}
    \begin{lstlisting}[firstnumber=7]
for $j\in 2^{\NN}$ with $j\leq n-1$
    \end{lstlisting}
    \label{alg:fastgreedy}
\end{algorithm}

The fast greedy delta distance is again dominated by the sorting of elements.
This time, $n\log n$ elements have to be sorted, thus its complexity is $n\log(n) \log(n\log n)=n \log(n)^2$.
Hence, the fast versions both have quasi linear runtime regarding length and linear runtime regarding dimensionality.

An inequality such as in Theorem~\ref{thm:fastdeltalb} does not exist for the fast greedy delta distance.
Also, there is no correlation between the (fast) delta distance and the (fast) greedy distance.
Though, the evaluation shows that the greedy delta distance provides a much better approximation in most cases.
Call for Section~\ref{sec:experiments} for an evaluation of their tightness to the congruence distance.

\section{Evaluation}
\label{sec:experiments}

Since the exact computation of the congruence distance is a computational hard problem (an thus not feasable in practical applications), we are mainly interested in the evaluation of the approximations.
Unfortunately, there is no direct algorithm for the computation of the congruence distance and we have to consider the computation of the congruence distance as a nonlinear optimization problem.
For two time series $S$ and $T$, we will denote the distance value computed by an optimizer with $d^O(S,T)$.
Since an optimizer might not find the global optimum, all values for the congruence distance (computed by an optimizer) in this section, are in fact upper bounds to the correct but unknown value of the congruence distance, i.\,e. $d^C(S,T)\leq d^O(S,T)$.
This given circumstance complicates the evaluation of our approximations to the congruence distance.

To estimate the tightness of the approximations, we first evaluate our optimizer on problems for which we know the correct results (cf.~Section~\ref{sec:evaloptimizer}).
In those cases, where the error of the optimizer is small, the estimation of the tightness of our approximations is accurate.
On the other hand, when the error of the optimizer increases, our estimation of the tightness of our approximations are loose and the approximation might be tighter than the experiments claim.

For a detailed explanation, consider a lower bound $\ell(S,T)$ for the congruence distance (e.\,g. $\ell$ might be one of $d^G$, $\tilde d^G$, $d^\Delta$, or $\tilde d^\Delta$) and suppose $d^C(S,T)=d^O(S,T)-\varepsilon$, i.\,e. $\varepsilon\geq 0$ is the error of the optimizer.
Then, we have the following correlation between the estimated tightness and the real tightness:
\begin{align*}
    \frac{\ell(S,T)}{d^C(S,T)} &= \frac{\ell(S,T)}{d^O(S,T)-\varepsilon} \geq \frac{\ell(S,T)}{d^O(S,T)}
\end{align*}
Hence, for small errors $\varepsilon$, the estimated tightness is accurate and for large errors $\varepsilon$ we underestimate the tightness.
In Section~\ref{sec:tightness}~and~\ref{sec:speedup}, we evaluate the tightness and the speedup of approximations to the (optimizer based) congruence distance, respectively.

\subsection{Congruence Distance: An Optimization Problem}
\label{sec:evaloptimizer}

Consider fixed time series $S$ and $T$ in $\RR^k$ with length $n$.
The congruence distance is a nonlinear optimization problem with equality based constraints.
The function to minimize is
\begin{align*}
    f\left( M,v \right) &= \sum_{i=0}^{n-1} \dE\left( s_i,M\cdot t_i+v \right)
\end{align*}
while the $k^2$ equality based constraints correspond the the constraints for orthogonal matrices:
\begin{align*}
    M\cdot M^T &= I.
\end{align*}
As a initial ``solution'' for the optimizer, we simply choose $M=I$ and $v=0$.

\begin{figure}
    \centering
    \includegraphics[width=.49\linewidth]{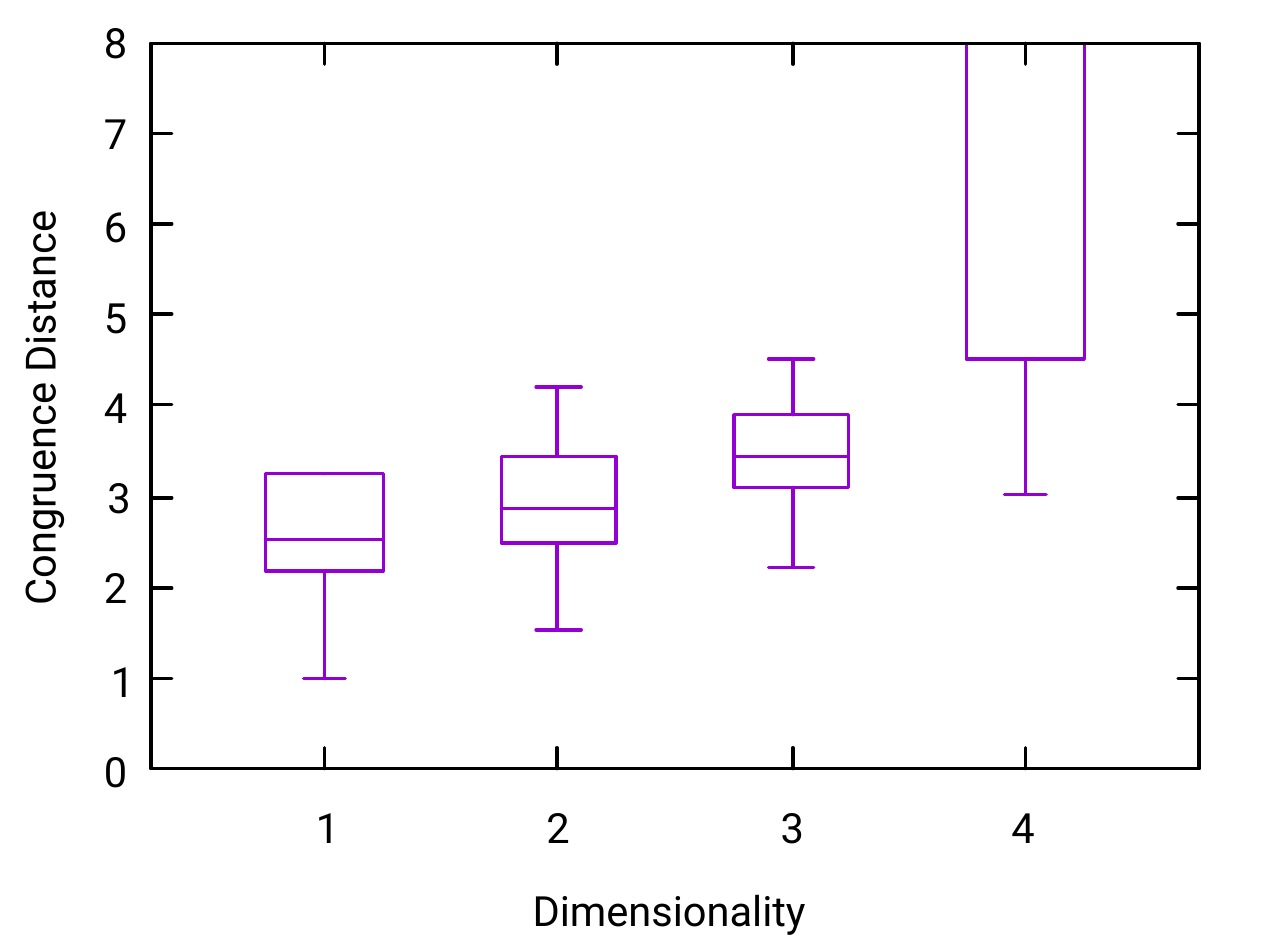}
    \includegraphics[width=.49\linewidth]{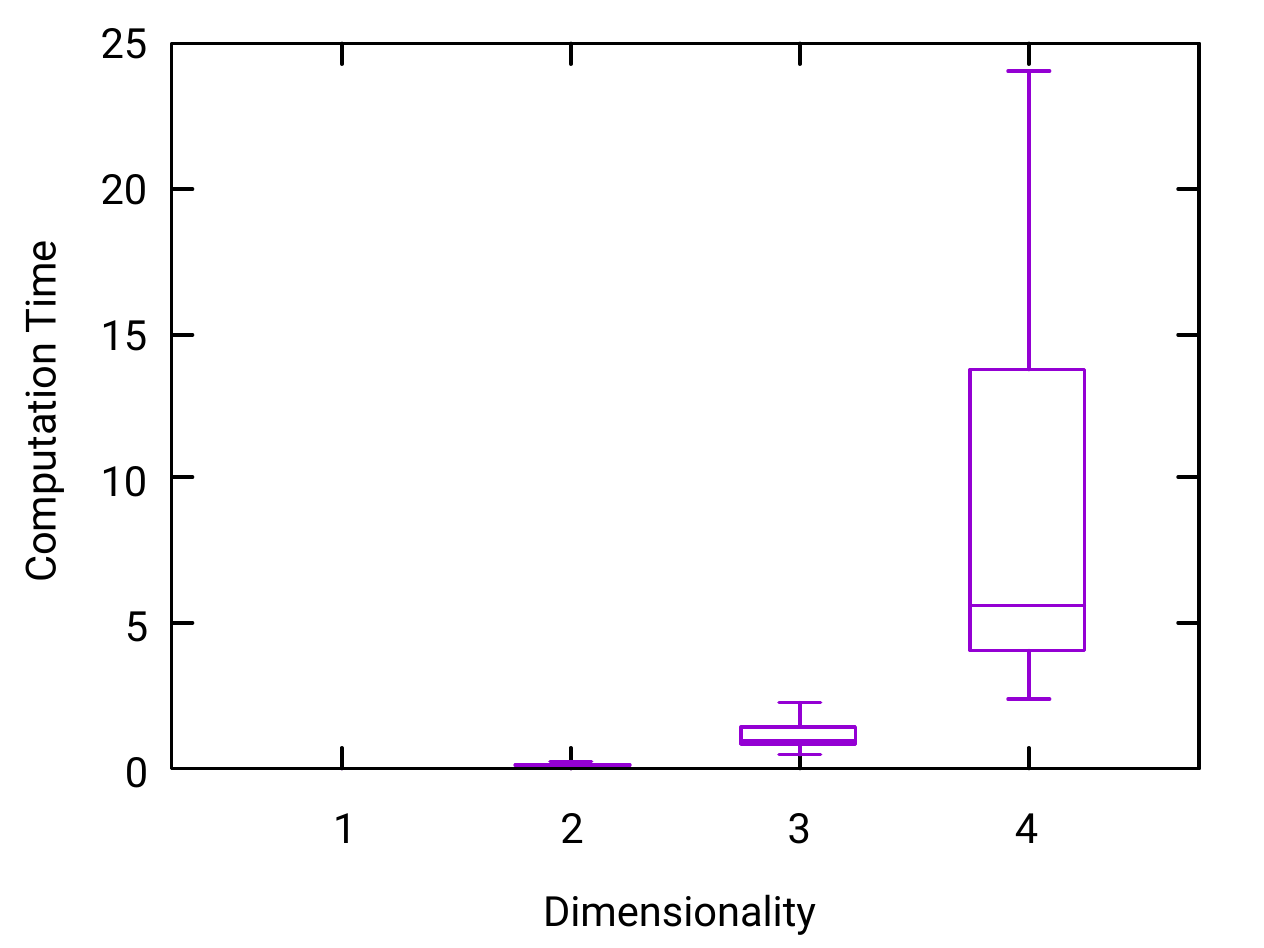}
    \caption{Boxplot: distance values (left) and runtimes (right) from our Optimizer on congruent time series.}
    \label{fig:zerodist}
\end{figure}

We manually transformed time series $T$ with a random orthogonal matrix $M^*$ and a small random vector $v^*$ and solved the optimization problem $d^C(T,M^*\cdot T+v^*)$ to examine whether our optimizer is working properly.
Clearly, we expect the optimizer to find a solution with value $0$.
Whenever the optimizer claimed large distance values, we concluded that the optimizer is not working.
We tried different optimizer strategies and chose an augmented lagrangian algorithm \cite{AUGLAGRANGE} with the BOBYQA algorithm \cite{BOBYCA} as local optimizer for further experiments because it promised the best performance with these experiments.
We used the implementations provided by the NLopt library \cite{NLOPT}.

We used the \ram{} dataset generator \cite{RAMgenerator} to evaluate the optimizer on time series with varying dimensionality.
We solved $400$ optimization problems for varying dimensionality and removed all of those runs where the optimizer did not find a reasonable solution (i.\,e. runs where the optimizer yielded solutions larger than $100$).
Figure~\ref{fig:zerodist} shows the distance values proposed by the optimizer (and therefore the error it makes) per dimensionality up to dimensionality $4$.
For higher dimensionalities, the optimizer completely failed to find any reasonable value near $0$ although we gave it enough resources of any kind (e.\,g. number of iterations, computation time, etc.).
Figure~\ref{fig:zerodist} also shows that the computation times rapidly increase with increasing dimensionality.
Because of the raising error and runtime with increasing dimensionality, an evaluation of the congruence distance on higher dimensionality is not feasable.
Hence, we can only consider up to $4$-dimensional time series in all further experiments.

\subsection{Tightness of Approximations}
\label{sec:tightness}

\begin{figure}
    \centering
    \includegraphics[width=.49\linewidth]{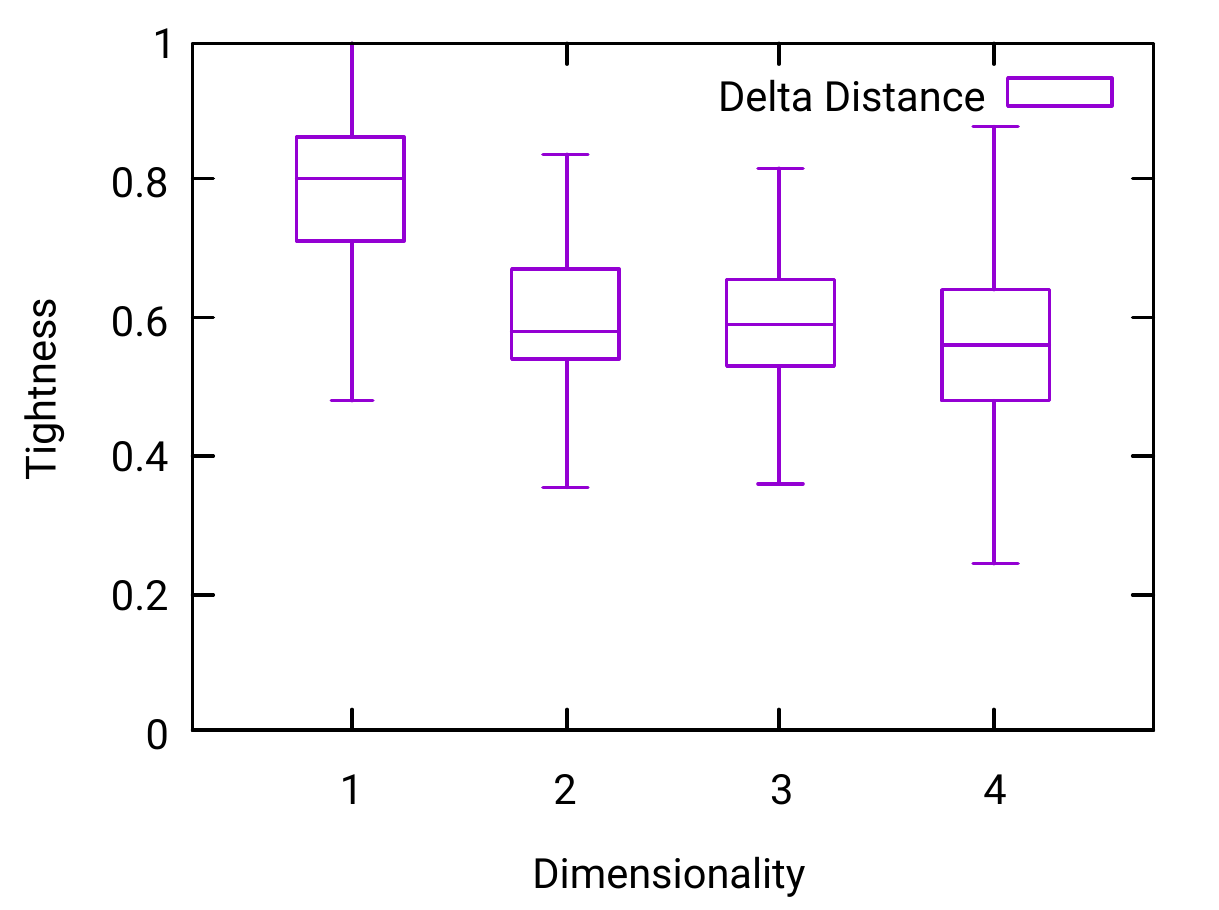}
    \includegraphics[width=.49\linewidth]{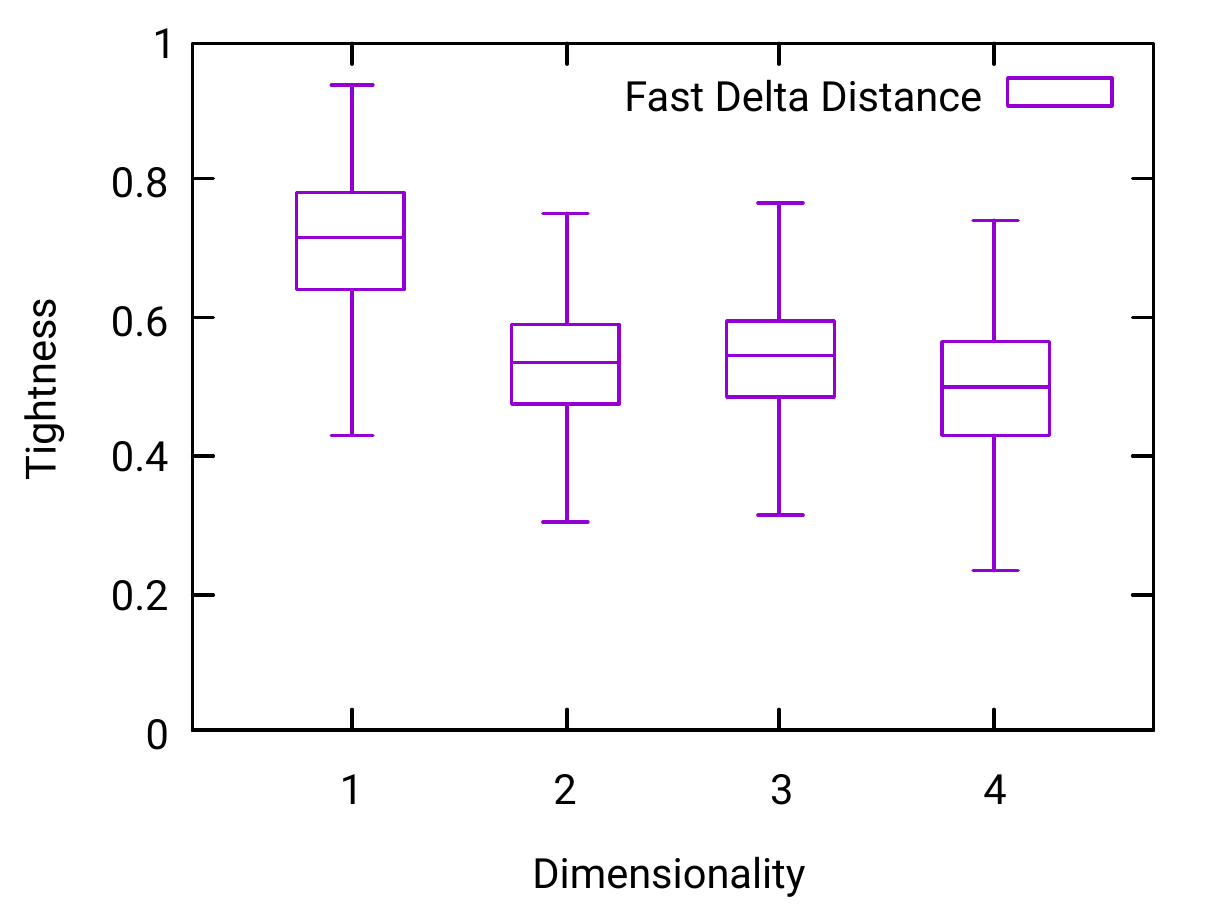}

    \includegraphics[width=.49\linewidth]{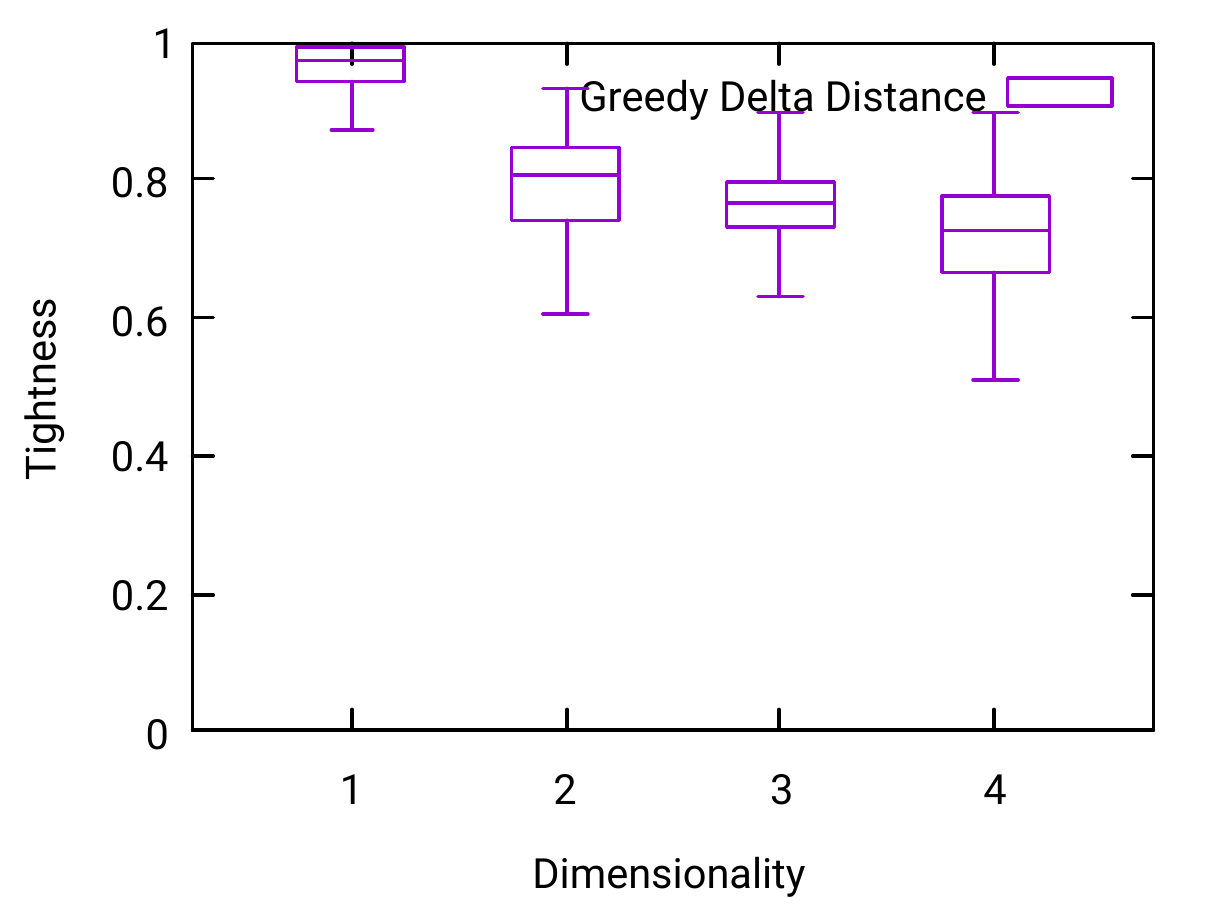}
    \includegraphics[width=.49\linewidth]{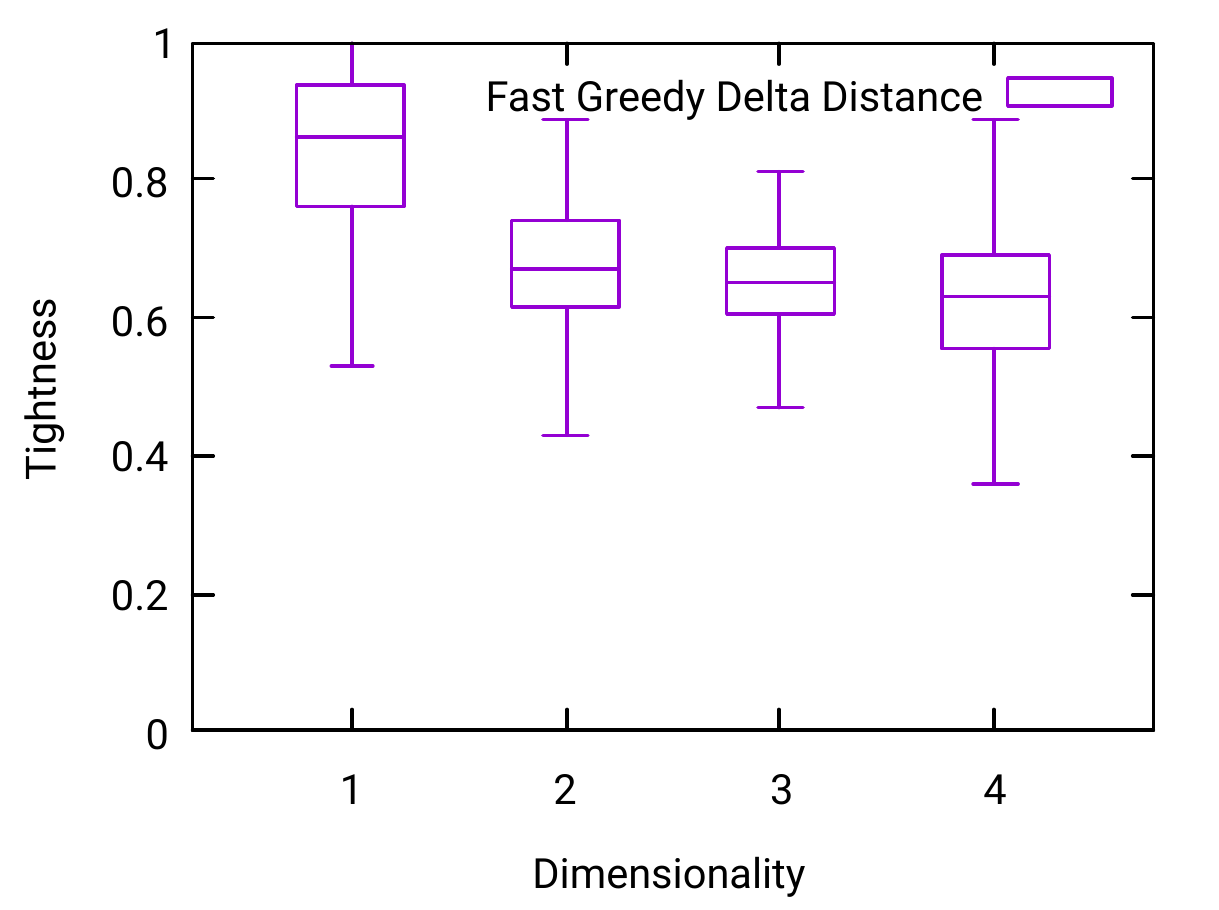}
    \caption{Average tightness of the delta distance (top left), the fast delta distance (top right), the greedy delta distance (bottom left), and the fast greedy delta distance (bottom right) to the congruence distance, respectively.}
    \label{fig:ramtightness}
\end{figure}

In order to evaluate the tightness of the (fast) delta distance and (fast) greedy delta distance as lower bounds to the congruence distance, we used the \ram{} dataset generator \cite{RAMgenerator} as well as a real world dataset with 2-dimensional time series (character trajectories \cite{UCIDatasets}, contains over 2800 time series).
Other real world datasets with higher dimensionality have not been suitable because the optimizer failed to compute the congruence distance.

Since making the (greedy) delta distance time warping aware is future work, we have to deal with time warping another way.
We simply preprocess our datasets, such that each time series, seen as a trajectory, moves with constant speed, i.\,e. for each \emph{dewarped} time series, the following holds:
\begin{align*}
    \dE\left( t_i,t_{i+1} \right) \approx \dE\left( t_{i+1},t_{i+2} \right).
\end{align*}
We achieve this property by simply reinterpolating the time series regarding the arc length.

Figure~\ref{fig:ramtightness} shows the tightness of the approximations on \ram{} datasets.
As we expected, the greedy delta distance provides the tightest approximation to the congruence distance (provided by our optimizer).

As we observed in Section~\ref{sec:evaloptimizer}, the error of our optimizer increases with increasing dimensionality.
Hence, the tightness of the optimizer to the real congruence distance is decreasing.
Since we can observe a similar behaviour here (the tightness of the approximation is decreasing with increasing dimensionality), the reason might be the inaccuracy of the optimizer.
Either way, we can see that the tightness is above $50\%$ in most cases.
Especially when using the greedy delta distance, the tightness is above $75\%$ in most cases.

On the character trajectories dataset, the delta distance and the greedy delta distance achieved a tightness of $63\%$ and $83\%$, respectively.

\subsection{Speedup of Approximations}
\label{sec:speedup}

\begin{figure}
    \centering
    \includegraphics[width=.49\linewidth]{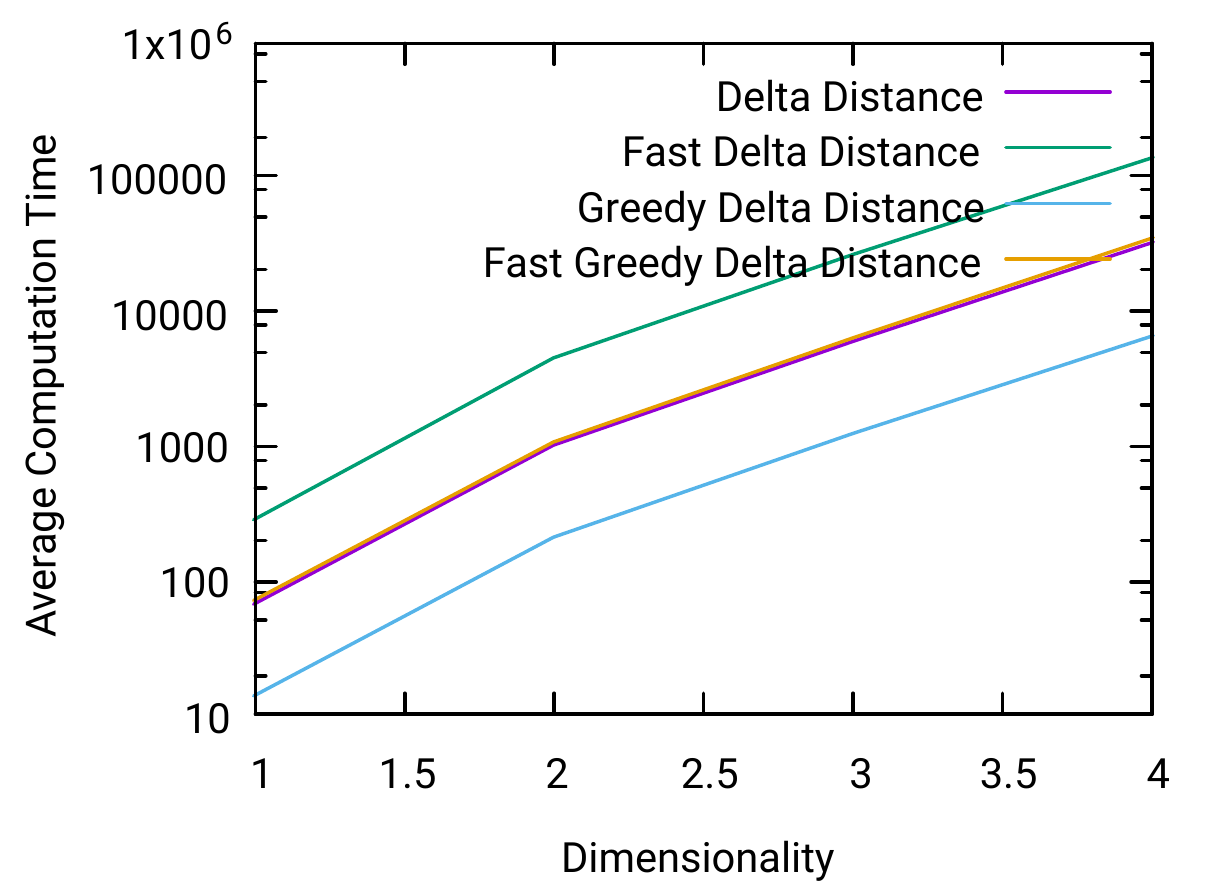}
    \caption{Average speedup of the approximations to our optimizer}
    \label{fig:ramspeedup}
\end{figure}

Figure~\ref{fig:ramspeedup} shows the speedup of the approximations to the optimizer.
As expected, the speedup increases exponentially with increasing dimensionality.
While the fast delta distance is the fastest algorithm, it also provides the worst approximation (compare with Figure~\ref{fig:ramtightness}).
On the other hand, the greedy delta distance provides the best approximation while being the slowest algorithm.
Still, the greedy delta distance is multiple orders of magnitudes faster than our optimizer.

The following speedups have been achieved on the character trajectory dataset:
$1642$ with the delta distance; $8040$ with the fast delta distance; $321$ with the greedy delta distance; $2287$ with the fast greedy delta distance.
The results are similar to those on the \ram{} generated datasets.

\section{Conclusion and Future Work}
\label{sec:conclusion}

In this paper, we analyzed the problem of measuring the congruence between two time series.
We provided four measures for approximating the congruence distance which are at least $2$ orders of magnitude faster than the congruence distance.
The first (namely, the delta distance) provides the additional ability to be used in metrix indexing structures.
The second (greedy delta distance) loses this benefit, but seems to achieve a better approximation.
Both approximations have linear complexity regarding the dimensionality but at least quadratic complexity regarding the length of the time series.
The other two approximations address this problem at a cost of approximation quality.
They have quasi-linear runtime regarding the length.

In practical applications, time series distance functions need to be robust against time warping.
The approximations provided in this work are based on comparing self-similarity matrices of time series.
Based on this idea, our next step is to develop a time warping distance function measuring the congruency.

\bibliographystyle{abbrv}
\bibliography{literature}

\end{document}